\documentclass[9pt,twocolumn,twoside]{IEEEtran}

\usepackage[utf8]{inputenc}
\usepackage[T1]{fontenc}
\usepackage{url}
\usepackage{ifthen}
\usepackage{cite}
\usepackage[cmex10]{amsmath} 
\usepackage{stackengine}
\usepackage{gensymb}
\usepackage{graphics}

\def\delequal{\mathrel{\ensurestackMath{\stackon[1pt]{=}{\scriptscriptstyle\Delta}}}}

\usepackage{mypackage}
\usepackage{mathrsfs}

\newtheorem{remark}{Remark}
\usepackage{mathtools}
\usepackage{tabu}

\DeclarePairedDelimiter\floor{\lfloor}{\rfloor}
\usepackage{float}

\newtheorem{theorem}{Theorem}

\usepackage{physics}
\usepackage{siunitx}

\newcommand{\sir}{\mathrm{SIR}}
\newcommand{\Pb}{\mathbb{P}}
\newcommand{\Eb}{\mathbb{E}}

\newcommand{\Lc}{\mathcal{L}}

\newcommand{\black}{\textcolor{black}}

\DeclareMathOperator*{\R}{\mathbb{R}}
\DeclareMathOperator*{\N}{\mathcal{N}}
\DeclareMathOperator*{\U}{\mathcal{U}}
\DeclareMathOperator*{\CN}{\mathcal{CN}}

\usepackage{multicol}

\usepackage[nomain,acronym,shortcuts]{glossaries}
\makeglossaries
\newcommand*{\acro}[3][]{\newacronym[#1]{#2}{#2}{#3}}

\acro{D2D}{device-to-device}
\acro{SIR}{signal-to-interference-ratio}
\acro{SINR}{signal-to-interference-plus-noise-ratio}
\acro{PCP}{Poisson cluster process}
\acro{CoMP}{coordinated multi-point}
\acro{BS}{base station} 
\acro{MD-CoMP}{macrodiversity CoMP transmission}
\acro{MAC}{medium-access-control}
\acro{JT-CoMP}{joint transmission CoMP}
\acro{CoMP-JT}{coordinated multipoint joint transmission}
\acro{SBS}{small base station}
\acro{MDSD}{multiple devices to the single device}
\acro{MDS}{maximum distance seperable}
\acro{SCN}{small cell network}
\acro{PPP}{Poisson point process}
\acro{TCP}{Thomas cluster process}
\acro{CSI}{channel state information}
\acro{PDF}{probability distribution function}
\acro{PMF}{probability mass function}
\acro{RV}{random variable}
\acro{i.i.d.}{independently and identically distributed}
\acro{MBMS}{multimedia broadcasting multicasting service}
\acro{EE}{energy efficiency}
\acro{HCP}{hard-core placement}
\acro{CCDF}{complementary cumulative distribution function}
\acro{CDF}{cumulative distribution function}
\acro{PC}{probabilistic caching}
\acro{RC}{random uniform caching}
\acro{CPF}{caching popular files} 
\acro{PGFL}{probability generating functional}
\acro{KKT}{Karush-Kuhn-Tucker}
\acro{PGF}{point generating function}
\acro{BPP}{binomial point process}
\acro{3GPP}{3rd Generation Partnership Project}
\acro{UAV}{unmanned aerial vehicle}
\acro{FAA}{Federal Aviation Administration}
\acro{UAS}{unmanned aircraft systems}
\acro{AP}{access point}
\acro{GSM}{Global System for Mobile Communications}
\acro{LTE}{long term evolution}
\acro{IoT}{Internet-of-things}
\acro{UE}{user equipment}
\acro{CNPC}{control and non-payload communication}
\acro{FHD}{full high definition}
\acro{NLoS}{non-line-of-sight}
\acro{LoS}{line-of-sight}
\acro{5G}{fifth-generation}
\acro{MGF}{moment generating function}
\acro{3D}{three-dimensions}
\acro{MIMO}{multiple-input multiple-output}
\acro{CB}{conjugate beamforming}
\acro{AU}{aerial user}
\acro{GU}{ground user}
\acro{CLT}{central limit theorem}
\acro{SE}{spectral efficiency}
\acro{MRT}{maximum ratio transmission}
\acro{ZFBF}{zero-forcing beamforming}
\acro{QoS}{quality of service}
\acro{GBS}{ground base station}
\acro{F-RAN}{fog-radio access network}
\acro{RAN}{radio access network}
\acro{RRH}{remote radio head}
\acro{mmWave}{millimeter wave}
\acro{C&C}{command and control}
\acro{ASE}{area spectral efficiency}
\acro{RWP}{random waypoint}
\acro{AUE}{aerial UE}
\acro{GUE}{ground UE}
\acro{SCDP}{successful content delivery probability}

\usepackage{amssymb}

\begin{document}

\newcommand\blfootnote[1]{%
  \begingroup
  \renewcommand\thefootnote{}\footnote{#1}%
  \addtocounter{footnote}{-1}%
  \endgroup
}


\title{Towards a Connected Sky: Performance of Beamforming with Down-tilted Antennas  for Ground and UAV User Co-existence}
\author{Ramy Amer, Walid Saad, and Nicola Marchetti}


 \author{Ramy Amer,~\IEEEmembership{Student~Member,~IEEE,} Walid~Saad,~\IEEEmembership{Fellow,~IEEE,} 
~and~Nicola~Marchetti,~\IEEEmembership{Senior~Member,~IEEE}
 
\thanks{Ramy Amer and Nicola~Marchetti are with CONNECT Centre for Future Networks, Trinity College Dublin, Ireland. Email:\{ramyr, nicola.marchetti\}@tcd.ie.}
\thanks{Walid Saad is with Wireless@VT, Bradley Department of Electrical and Computer Engineering, Virginia Tech, Blacksburg, VA, USA. Email: walids@vt.edu.}
\thanks{This publication has emanated from research conducted with the financial support of Science Foundation Ireland (SFI) and is co-funded under the European Regional Development Fund under Grant Number 13/RC/2077, and the U.S. National Science Foundation under Grants CNS-1836802 and IIS-1633363.}
} 
\maketitle

\begin{abstract}
Providing connectivity to aerial users (AUs) such as cellular-connected unmanned aerial vehicles (UAVs) is a key challenge for tomorrow's cellular systems. In this paper, the use of conjugate beamforming (CB) for simultaneous content delivery to an AU co-existing with multiple ground users (GUs) is investigated. In particular, a content delivery network of uniformly distributed massive multiple-input multiple-output (MIMO)-enabled ground base stations (BSs) serving both aerial and ground users through spatial multiplexing is considered. For this model, \black{the successful content delivery probability (SCDP) is derived} as a function of the system parameters. The effects of various system parameters such as antenna down-tilt angle, AU's altitude, number of scheduled users, and number of antennas on the achievable performance are then investigated. Results reveal that whenever the AU's altitude is below the BS height, the antennas' down-tilt angles yield an inherent tradeoff between the performance of the AU and the GUs. However, if the AU's altitude exceeds the BS height, down-tilting the BS antennas with a considerably large angle improves the performance of both the AU and the GUs.
 
%
%
%
%
%
%
%
%
%
%
\end{abstract}
\begin{IEEEkeywords}
Cellular-connected UAVs, conjugate beamforming.	
\end{IEEEkeywords}
\vspace{-0.4 cm}
\section{Introduction}	
\black{A tremendous increase in the use of \acp{UAV} in a wide range of applications, ranging from airborne \acp{BS}, delivery of goods, to traffic control, is expected in the foreseeable future \cite{8660516,mozaffari2016unmanned,eldosouky2019drones,TUAVmag2019}. To enable these applications, UAVs must communicate with one another and with ground devices. To enable such communications, it is imperative to connect UAVs, seen as \acp{AU}, to the ubiquitous wireless cellular network. Such cellular-connected UAVs have recently attracted attention in cellular network research in both academia and industry \cite{azari2017coexistence,lin2018sky,85317111,8528463,amer2020caching,comp-meet-uavs} due to their ability to pervasively communicate. However, cellular networks have been designed to provide connectivity to \acp{GU} rather than \acp{AU} \cite{azari2017coexistence}. For instance, cellular-connected UAV communication possesses substantially different characteristics that pose new technical challenges which include: dominance of \ac{LoS} interference and reduced \acp{GBS} antenna gain \cite{azari2017coexistence}.}


In this regard, the authors in \cite{azari2017coexistence} studied the feasibility of supporting drone operations using existing cellular infrastructure. Results revealed that the favorable propagation conditions that \acp{AU} enjoy due to their altitude is also one of their strongest limiting factors since they are susceptible to \ac{LoS} interference. Meanwhile, the authors in \cite{85317111} minimized the UAV's mission completion time by optimizing its trajectory while maintaining reliable communication with the \acp{GBS}. In \cite{8528463}, through system simulations, the authors evaluated the performance of the downlink of \acp{AU} when supported by either a traditional cellular network, or a massive \ac{MIMO}-enabled network \black{with \ac{ZFBF}}. In \cite{amer2020caching}, the authors showed that cooperative transmission significantly improves the coverage probability for high-altitude \acp{AU}. However, while the works in \cite{azari2017coexistence,85317111}, \cite{amer2020caching}, and \cite{comp-meet-uavs} have analyzed the performance of cellular-connected UAVs, their approaches can not be used to effectively improve the performance of AUs while enhancing \ac{SE} by spatial multiplexing. Also, even though the work in \cite{8528463} has proposed MIMO beamforming for an AU co-existing with multiple GUs, this work provides no analytical characterization of the performance of AUs or the impact of the antennas' down-tilt angles.

\black{The main contribution of this paper is a comprehensive analysis of cellular communications with AUs. In particular, we propose a MIMO \textit{\ac{CB}} approach that can improve the performance of cellular communication links for the \acp{AU} and effectively enhance the cellular system \ac{SE}. We consider a network of one \ac{AU} co-existing with multiple \acp{GU} that are being simultaneously served via massive MIMO-enabled \acp{GBS}. We introduce a novel analytical framework that can be leveraged to characterize the performance of the spatially multiplexed \ac{AU} and \acp{GU}. Given the different channel characteristics and the corresponding precoding vectors among \acp{GU} and the \ac{AU}, we first derive the gain of intended and interfering channels for both kind of users. We then analytically characterize the \ac{SCDP} as a function of the system parameters. \emph{To  our best knowledge, this is the first work to perform a rigorous analysis of MIMO \ac{CB} to simultaneously serve aerial and ground users.}}

\vspace{-0.3 cm}
\section{System Model}
Consider a cellular network composed of massive \ac{MIMO}-enabled \black{\acp{BS} $b_i$ distributed according to a homogeneous \ac{PPP} $\Phi$ of intensity $\lambda$, where $\Phi=\{ b_i \in \mathbb{R}^2, \forall i \in \mathbb{N}^+ \}$}. A three-sectored cell is associated with each \ac{BS}, with each sector spanning an angular interval of $120\,^{\circ}$. Each sector has a large antenna array of $M$ antennas  
%
%
each of which has a horizontal constant beamwidth of $120\,^{\circ}$, and vertical beamwidth $\theta_B$. \ac{CB} is employed to simultaneously serve a selected set $\mathcal{K}$ of $K$ users. These $K$ users are viewed as an \ac{AU} that is scheduled with a set $\mathcal{K}_{G}$ of $K-1$  \acp{GU}, as done in \cite{8528463}. This assumption is in line with the fact that the number of current \acp{GU} is much larger than the number of \acp{AU}. \black{We assume that the \acp{GU} are located according to some independent stationary point process.} BSs are deployed at the same height $h_{{\rm BS}}$ while \acp{AU} and \acp{GU} are at altitudes $h_d$ and $h_g$, respectively, where $h_d \gg h_g$. Given the symmetry of the problem, we consider the performance of the typical ground and aerial users located at $(0,0,h_g)$, and $(0,0,h_d)$, respectively. \black{We also refer to the serving \ac{BS} as \emph{tagged \ac{BS}}, which is the nearest BS to the origin $(0,0)\in \R^2$, with $d_{ig}$ and $d_{id}$ being the distances from the \ac{GBS} to the typical GU and AU, respectively.}
%
%
%


For \acp{GU}, we consider \ac{i.i.d.} quasi-static Rayleigh fading channels. The channel vector between the $M$ antennas of tagged BS $i$ and \ac{GU} $k$ is $\sqrt{\beta_{ik}}\boldsymbol{h}_{ik}$, where $\boldsymbol{h}_{ik} \sim \CN(\boldsymbol{0},\sigma^2\boldsymbol{I}_M)$ for $k \in \mathcal{K}_{G}$. $\sigma^2$ is the channel variance between each single antenna and user $k$, and $\boldsymbol{I}_M$ is the $M\times M$ identity matrix. \black{$\beta_{ik}=d_{ig}^{-\alpha}$ defines the large-scale channel fading}. We also assume that the \ac{GU} channels are dominated by \ac{NLoS} transmission.  
For the \ac{AU}, we assume a wireless channel that is characterized by both large-scale and small-scale fading. For the large-scale fading, the channel between \ac{BS} $i$ and the \ac{AU} includes \ac{LoS} and \ac{NLoS} components, which are considered separately along with \black{their probabilities of occurrence \cite{8713514}}. 
For small-scale fading, we adopt a Nakagami-$m_v$ model for the channel between each  single antenna and the \ac{AU}, \black{as done in \cite{amer2020caching,8713514,7967745},} with the following \ac{PDF}:
\begin{align}
\label{Naka}
f_{\Omega_v}(\omega,\eta) = \frac{2(\frac{m_v}{\eta})^{m_v} \omega^{2m_v-1}}{\Gamma(m_v)} {\rm exp}\big(-\frac{m_v}{\eta}\omega^2\big), 
\end{align}
where $v\in \{l,n \}$, $m_l$ and $m_n$ are the fading parameters for the \ac{LoS} and NLoS links, respectively, with $m_l>m_n$, and $\eta$ is a controlling spread parameter. When $m_v=\eta = 1$, Rayleigh fading is recovered with an exponentially distributed instantaneous power, which is the case for \acp{GU} or \acp{AU} with no \ac{LoS} communication. For Nakagami channels, we assume that the phase $\theta_{ng}$ is uniformly distributed in $[0,2\pi]$, i.e., $\theta_{ng} \sim \U(0,2\pi)$.	
Given that $\omega \sim $ Nakagami$(m_v,\eta)$, it directly follows that the channel gain $\omega^2 \sim \Gamma(m_v,\frac{\eta}{m_v})$. 
%
%

3D blockage is characterized by the fraction $a$ of the total land area occupied by buildings, the mean number of buildings being $\nu$  per \SI{}{km}$^2$, and the buildings' height modeled by a Rayleigh \ac{PDF} with a scale parameter $c$. Hence, the probability of  \ac{LoS} when served from \ac{BS} $i$, at a horizontal-distance $r_{i}$ from the typical \ac{AU}, is given as \cite{azari2017coexistence}:
\begin{align}
\label{los-eqn}
\Pb_{l}(r_i) = \prod_{n=0}^{p}\Big[1 - {\rm exp}\big(- \frac{\big(h_{\textrm{BS}} + \frac{h(n+0.5)}{p+1}\big)^2}{2c^2}\big) \Big], 
\end{align}
where $h=h_d - h_{\textrm{BS}}$ and $p=\floor{\frac{r_i\sqrt{a\nu}}{1000}-1}$. In our model, we assume that the AUs are deployed in an urban environments, where $a$ and $\nu$ take relatively large values. \black{Therefore, the large-scale channel fading for the \ac{AU} is given by $d_{id}^{-\alpha_v}$, where $v\in \{l,n \}$, $\alpha_{l}$ and $\alpha_{n}$ are the path loss exponents for \ac{LoS} and NLoS links, respectively, with $\alpha_{l}<\alpha_{n}$.}

\black{For a general user $k\in \mathcal{K}$ at altitude $h_k \in \{h_d,h_g\}$, the antenna directivity gain can be written similar to \cite{azari2017coexistence} as $G(r_i) =G_m$, for $r_i \in \mathcal{S}_{bs}$, and $G_s$, for $r_i \notin \mathcal{S}_{bs}$, where $r_i$ is the horizontal-distance to the BS, $\mathcal{S}_{bs}$ is formed by all the distances $r_{i}$ satisfying $h_{\textrm{\ac{BS}}} - r_{i} {\rm tan}(\theta_t + \frac{\theta_B}{2}) < h_k < h_{\textrm{\ac{BS}}} - r_{i} {\rm tan}(\theta_t - \frac{\theta_B}{2})$, and $\theta_t$ and $\theta_B$ denote respectively the antenna down-tilt and beamwidth angles.}
%
Therefore, the antenna gain plus path loss will be	
\begin{align}
\zeta_v(r_i) = A_v G(r_i) d_i^{-\alpha_v}  = A_v G(r_i) \big(r_i^2 + (h_k-h_{\textrm{\ac{BS}}})^2\big)^{-\alpha_v/2},
\nonumber 
\end{align}
\black{where $d_i \in \{d_{ig},d_{id}\}$, $v\in \{l,n \}$, and $A_{l}$ and $A_{n}$ are the path loss constants at a reference distance $d_i = \SI{1}{m}$ for \ac{LoS} and \ac{NLoS}, respectively.} \black{For the typical \ac{GU}, $d_i=d_{ig}$, $h_k=h_g$ and, by \ac{NLoS} assumption, $v=n$.}
\black{Note that, since one \ac{AU} is simultaneously scheduled with $K-1$ \acp{GU}, the $K$ scheduled users encounter independent  small-scale fading. Also, for the $K-1$ \acp{GU}, the small-scale fading is \ac{i.i.d.} Moreover, for the AU, the impact of the channel spatial correlation can be significantly reduced using effective MIMO antenna design techniques, e.g., using antenna arrays whose elements have orthogonal polarizations or patterns \cite{bhagavatula2008performance}. Therefore, for analytical tractability, we ignore such spatial correlation.}


Next, we introduce our proposed \ac{CB} framework to spatially multiplex one \ac{AU} and $K-1$ \acp{GU}. We develop a novel mathematical framework that can be leveraged to characterize the performance of aerial and ground users. This, in turn, allows us to quantify the impact of different system parameters, on the performance of  AUs and GUs. 

\vspace{-0.25 cm}
\section{Content Delivery Analysis}		
We assume that perfect \ac{CSI} is available at the tagged \ac{BS}. Linear precoding in terms of \ac{CB} creates a $K\times1$ transmission vector $\boldsymbol{s}'$ for $M$ antennas by multiplying the original data vector $\boldsymbol{s}$ by a precoding matrix $\boldsymbol{W}$:  $\boldsymbol{s}'=\boldsymbol{W}\cdot \boldsymbol{s}$, where $[\boldsymbol{W}]_{M\times K}$ consists of the beamforming weights. 
\black{Let $\boldsymbol{H}$ be the $M\times K$ channel matrix between $M$ antennas of the tagged \ac{BS} $i$ and its $K$ scheduled users, written as $\boldsymbol{H}_i = \Big[ \boldsymbol{h}_{i1} \dots \boldsymbol{h}_{ik} \dots \boldsymbol{h}_{iK} \Big]$, 
 where $\boldsymbol{H}_{i} \in \mathbb{C}^{M\times K}$, and $\boldsymbol{h}_{ik} \in \mathbb{C}^{M\times1}$.} 
For \ac{CB}, tagged \ac{BS} $i$ creates a precoding matrix $\boldsymbol{W}_i = \Big[ \boldsymbol{w}_{i1} \dots \boldsymbol{w}_{ik} \dots \boldsymbol{w}_{iK} \Big]$, with $\boldsymbol{w}_{ik} = \frac{\boldsymbol{h}_{ik}^{H}}{\lVert\boldsymbol{h}_{ik}\rVert}$, where each beam is normalized to ensure equal power assignment \cite{6823643}.		
Moreover, let $\boldsymbol{f}_{jk}$ be the interfering channel between interfering BS $j$ and typical user $k$. Neglecting thermal noise, the received signal at scheduled user $k$, denoted as $y_{ik}$, is given by
\begin{align}			
\label{rec-pwr} 
P(r_i) \boldsymbol{h}_{ik} \boldsymbol{w}_{ik} 
 s_{ik} + \sum_{\kappa\in \mathcal{K}_{G}} P(r_i) \boldsymbol{h}_{ik} \boldsymbol{w}_{i\kappa} s_{i\kappa} +
 \sum_{j\in\Phi^o } \sum_{l=1}^{K} P(u_j) \boldsymbol{f}_{jk} \boldsymbol{w}_{jl} s_{jl},  
 \nonumber 
\end{align}				
where $\Phi^o=\Phi\setminus{\{i\}}$. \black{The first term in the above equation represents the desired signal from tagged \ac{BS} $i$ with $P(r_i) =\sqrt{\frac{P_t}{K}}  \zeta_v(r_i)^{0.5}$ representing the received power and $P_t$ denoting the BS transmission power. The second and third terms represent the intra- and inter-cell interference, denotes as $I_{\text{in}}$ and $I_{\text{out}}$, respectively.} The information signal intended for user $k$ is denoted by a complex scalar $s_{ik}$ with unit average power, i.e., $\Eb[|s_{ik}|^2]=1$. 

Since we assume both \ac{LoS} and \ac{NLoS} communications for the \ac{AU}, with corresponding small-scale fading, we need to distinguish between the two communication paradigms. For the \ac{NLoS} case, all the $K$ users experience Rayleigh small-scale fading. For \ac{LoS} communication, however, only the \ac{AU} experiences Nakagami-$m_l$ small-scale fading, where $m_l>1$. We hence start by characterizing the gain of intended and interfering channels in Table \ref{ch5:table1}.

The second and third columns in Table \ref{ch5:table1} list the marginal channel distributions, i.e., the channel from each single antenna to its intended receiver. We also use interfering BSs to refer to either intra- or inter-cell BS. The first row in Table \ref{ch5:table1} represents the intended channel gain for \acp{GU}. It is shown that the equivalent channel gain from tagged BS to its associated GU follows $\Gamma(M,\sigma^2)$ \cite{6823643}. Similarly, the second row represents the intended channel gain for the \ac{AU},  which is the sum of $M$ independent \acp{RV}, each of which follows $\Gamma(m_v,\frac{\eta}{m_v})$. Hence, its intended channel gain follows $\Gamma(m_vM,\frac{\eta}{m_v})$. 
%
%
The third row stands for the interference power caused by transmission of a single beam from an interfering BS to its associated \ac{GU} when seen by the typical \ac{GU}, which follows $\Gamma(1,\sigma^2)$ \cite{6823643}.  The fourth (fifth) row describes cases in which a single beam from an interfering BS to its associated \ac{GU} (AU) is transmitted and seen by the typical \ac{AU} (GU). \black{Similarly, the sixth row describes cases in which a single beam from an interfering BS to its associated \ac{AU} is transmitted and seen by the typical \ac{AU}}. 
Next, we derive the channel gain for the fourth case, whereas the fifth and sixth cases follow in the same way and are omitted due to space limitations.
\begin{table}
\vspace{-0.5 cm}
\caption{Channel gains for intended and interfering links.}
\vspace{-0.2 cm}
\begin{center}
\scalebox{0.8}{
  \begin{tabular}{ |p{.15cm}|p{1.9cm}|p{2.1cm}|p{0.4cm}|p{0.9cm}|p{1.5cm}|}  
    \hline
\textbf{No} & \textbf{Precoding for channel}&\textbf{Traverse through channel}& \textbf{Seen by}& \textbf{Intended}& \textbf{Channel gain}\\ \hline
1&$\CN(0,\frac{\sigma^2}{2})$&$\CN(0,\frac{\sigma^2}{2})$   &GU &Yes & $\Gamma(M,\sigma^2)$\\ \hline
2&Nakagami$(m_v,\eta)$&Nakagami$(m_v,\eta)$    & AU &Yes    & $\Gamma(m_vM,\frac{\eta}{m_v})$ \\ \hline 
3&$\CN(0,\frac{\sigma^2}{2})$&$\CN(0,\frac{\sigma^2}{2})$&GU &No& $\Gamma(1,\sigma^2)$\\ \hline	
4&$\CN(0,\frac{\sigma^2}{2})$&Nakagami$(m_v,\eta)$& AU &No & $\Gamma(1,\eta)$\\ \hline		 
5&Nakagami$(m_v,\eta)$&$\CN(0,\frac{\sigma^2}{2})$  &GU &No  & $\Gamma(1,\sigma^2)$\\ \hline		 
6&Nakagami$(m_v,\eta)$&Nakagami$(m_v,\eta)$    & AU &No & $\Gamma(1,\eta)$ \\ 
\hline
\end{tabular}}
\label{ch5:table1}
\end{center}
			\vspace{-0.50 cm} 			
\end{table}

\begin{theorem}
\label{ch5:theorem1}
\black{Under the massive MIMO assumption, whenever a single beam from an interfering BS is received by the typical AU then, the interference channel gain will be given by $\Gamma(1,\eta)$.}
\end{theorem}
\begin{proof}
We write the interfering channel coefficient as
\begin{align}
h_j &= \boldsymbol{w}_{j\kappa} \boldsymbol{f}_{jk} =  \frac{\boldsymbol{h}_{j\kappa}^{H} \boldsymbol{f}_{jk} }{\big\lVert\boldsymbol{h}_{j\kappa}\big\rVert} 
\delequal \frac{\sum_{o=1}^{M} X_o \times Y_o}{\sqrt{\sum_{q=1}^{M} Z_q}}
 \\
&\overset{(a)}{\delequal} \frac{\sum_{o=1}^{M} X_o \times Y_o}{\sqrt{W}}
\overset{(b)}{\delequal} \frac{\sum_{o=1}^{M} X_o \times Y_o}{Q},
\label{interfer_pwr}
\end{align}
where $X_o\sim \CN(0,\frac{\sigma^2}{2})$, $Y_o \sim {\rm Nakagami}(m_v,\eta)$, $Z_q \sim {\rm exp}(\frac{1}{\sigma^2})$, $W\sim\Gamma(M,\frac{1}{\sigma^2})$, and $Q\sim{\rm Nakagami}(M,\frac{M}{\sigma^2})$; (a) follows since $W$ is a sum of $M$ \ac{i.i.d.} exponential \acp{RV}, hence it follows $\Gamma(M,\frac{1}{\sigma^2})$. (b) follows since $Q$ equals the square root of the \ac{RV} $W\sim\Gamma(M,\frac{1}{\sigma^2})$, hence $Q$ follows ${\rm Nakagami}(M,\frac{M}{\sigma^2})$.  
Denoting the numerator of $h_j$ as $z$, and writing $z$ as sum of real and imaginary \acp{RV}:
\begin{align}
 \Re(z) &\delequal \sum_{o=1}^{M}  \Big(  \underbrace{X_o {\rm cos}(\theta_{ng_o})- X_o {\rm sin}(\theta_{ng_o}) }_{\text{RV}\#1}\Big)	
 \cdot \underbrace{Y_o}_{\text{RV}\#2} , 	
\end{align}		
where, by assumption, $\theta_{ng_o} \sim \U(0,2\pi)$. We hence have a sum of $M$ \ac{i.i.d.} \acp{RV}, each of which is the product of two independent \acp{RV} whose means and variances are $\{\mu_1,\mu_2\}$ and $\{\sigma_1^2,\sigma_2^2\}$, respectively. It can easily be shown that $\mu_1=0$ and $\sigma_1^2=\frac{\sigma^2}{2}$. \black{For large $M$,} using the \ac{CLT}, we approximate the \ac{PDF} of $\Re(z)$ to $\N(\mu_{12},\sigma_{12}^2)$, whose mean and variance are respectively $\mu_{12} \overset{}{=}  \mu_1 \mu_2 = 0$, and $\sigma_{12}^2$
\begin{align}
&\overset{}{=} \sigma_1^2 \sigma_2^2 + \sigma_1^2 \mu_2^2 + \mu_1^2 \sigma_2^2 
\nonumber \\
& \overset{(a)}{=}\frac{\sigma^2}{2} \eta \Big( 1 - \frac{1}{m_v}\Big(\frac{\Gamma(m_v+0.5)}{\Gamma(m_v)}\Big)^2\Big)  +
\frac{\sigma^2}{2} \Big( \frac{\Gamma(m_v+0.5)}{\Gamma(m_v)} \Big(\frac{\eta}{m_v}\Big)^{0.5}\Big)^2
\nonumber \\
& \overset{}{=} \frac{\sigma^2\eta}{2} - \frac{\sigma^2\eta}{2m_v} \Big(\frac{\Gamma(m_v+0.5)}{\Gamma(m_v)}\Big)^2  + \frac{\sigma^2\eta}{2 m_v} \Big(\frac{\Gamma(m_v+0.5)}{\Gamma(m_v)}\Big)^2
 \overset{}{=}  \frac{\sigma^2\eta}{2},
\end{align}
where (a) follows from the mean and variance formulas for Nakagami$(m_v,\eta)$. For the dominator of (\ref{interfer_pwr}), we use the Stirling approximation  to approximate the \ac{PDF} of $Q$ by	
\begin{align}
f_{\Omega}(\omega,M,M/\sigma^2) = \frac{1}{\omega}\big(\frac{ \omega^2}{\frac{M}{\sigma^2} e^{\frac{\omega^2}{M/\sigma^2}-1}}\big)^M. 
\end{align}
The fraction raised to the $M$-th power is smaller than one, and the integral is one (since it is a \ac{PDF}). In fact, the factor raised to the $M$-th power is one only when $\omega=\frac{\sqrt{M}}{\sigma}$. \black{Hence, for large $M$,}  from the \ac{CLT}, $\Re(h_j) \sim \frac{\N\big(0,\frac{\sigma^2\eta}{2M} \big)}{\frac{\sigma}{\sqrt{M}}} \delequal  \N(0,\frac{\eta}{2} )$. Similarly, $\Im(h_j) \sim \N(0,\frac{\eta}{2} )$. Hence, the channel gain $|h_j|^2 = \big(\sqrt{\Re{h_j}^2 + \Im{h_j}^2}\big)^2 \sim \Gamma(1,\eta)$. This completes the proof.
\end{proof}

%
 %
%
%
%
%
Next, we derive the \ac{SCDP} for the \ac{AU}, which is defined as the probability of obtaining a requested content with \ac{SIR} higher than a target \ac{SIR} $\vartheta$. This is an important performance metric that is widely studied in content delivery and caching networks \cite{8412262} and \cite{chaccour2019reliability}. The same methodology can be applied to obtain that for \acp{GU}, but the details are omitted due to space constraints. 
%
%
%
%
We next index the \ac{AU} as $k=1$. Let $h_{iK}=\sum_{\kappa\in{\mathcal{K}_{G}}} |\boldsymbol{w}_{i\kappa} \boldsymbol{f}_{i1}|^2$ denote the intra-cell interference power. From Theorem \ref{ch5:theorem1}, $|\boldsymbol{w}_{i\kappa} \boldsymbol{f}_{i1}|^2\sim\Gamma(1,\eta)$. Neglecting the spatial correlation, we have $h_{iK}$ representing  sum of $K-1$ Gamma \acp{RV}, which yields $h_{iK} \sim  \Gamma(K-1,\eta)$. Similarly, the inter-cell interference power $h_{jK}=\sum_{l=1}^{K}|\boldsymbol{w}_{il} \boldsymbol{f}_{j1}|^2\sim\Gamma(K,\eta)$. \black{Finally, according to the void probability of PPPs \cite{haenggi2012stochastic}, the \ac{PDF} of the horizontal-distance $r$ to the \textit{tagged BS} is $f_{R}(r)=2\pi\lambda r e^{-\pi\lambda r^2}$.}
\begin{theorem}
\label{theorem2} 
The \black{unconditional} \ac{SCDP} for the \ac{AU} is given by 
\begin{align} 
\label{cov-prob}
\Pb_{{\rm c}} &= \Pb\big(\sir > \vartheta \big)
= \int_{r=0}^{\infty} \Big[ \Pb_{{\rm c}|r}^l \Pb_{l}(r) + \Pb_{{\rm c}|r}^n \Pb_{n}(r)
\Big] f_{R}(r) \dd{r},
\end{align} 
where $\Pb_{{\rm c}|r}^v = \lVert e^{\boldsymbol{T}_{M_v}}\rVert_1$, 
 $\lVert.\rVert_1$ denotes the induced $\ell_1$ norm, and $\boldsymbol{T}_{M_v}$ is the lower triangular Toeplitz matrix of size $M_v\times M_v$: 
\[ \boldsymbol{T}_{M_v} =
\begin{bmatrix}
    t_{0}       \\
    t_{1}       & t_{0}  \\
    \vdots &  \vdots & \ddots  \\
    t_{M_v-1}       & \dots & t_{1}  & t_{0} 
\end{bmatrix};
\]
where  $M_v=Mm_v$, and \black{the non-zero entries for row $i$ and column $j$ are $t_{i-j} =\frac{(-s_v)^{i-j}}{(i-j)!} \varpi^{(i-j)}(s_v)$}; $s_v = \frac{\vartheta K m_v}{\eta P_t \zeta_v(r)}$,  
$\varpi(s_v)=- (K-1) {\rm log}(1 + s_v \eta P_v(r)^2) -  2\pi  \lambda \int_{\nu=r}^{\infty} \big(1 -
\Pb_{l}(\nu)\delta_l(\nu,s_v) - \Pb_{n}(\nu)\delta_n(\nu,s_v) \big)\nu\dd{\nu}$, 
and $\varpi^{(k)(s_v)}= \frac{d^k}{d s_v^k} \varpi(s_v)$; $\delta_l(\nu,s_v)=\big(1 + s_v \eta P_l(\nu)^2\big)^{-K}$, and $\delta_n(\nu,s_v)=\big(1 + s_v \eta P_n(\nu)^2 \big)^{-K}$.
\end{theorem}

\begin{figure}[!t]
    \vspace{-0.4cm}
    \centering
    \subfigure[Number of antennas $M=4$]
    {
        \includegraphics[width=1.6in]{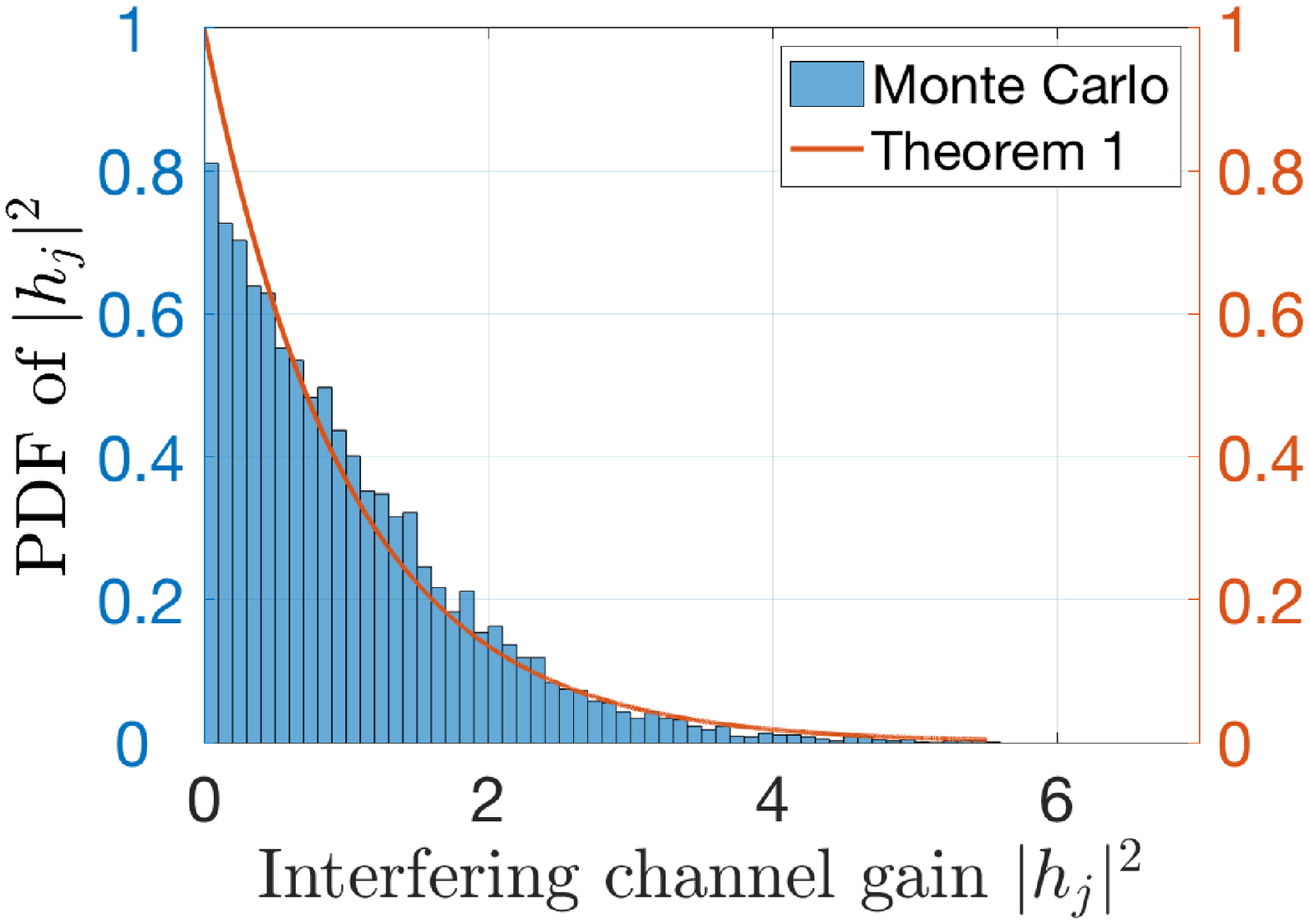}		
        \label{pwr1}
    }
    \subfigure[Number of antennas $M=32$]
    {
        \includegraphics[width=1.6in]{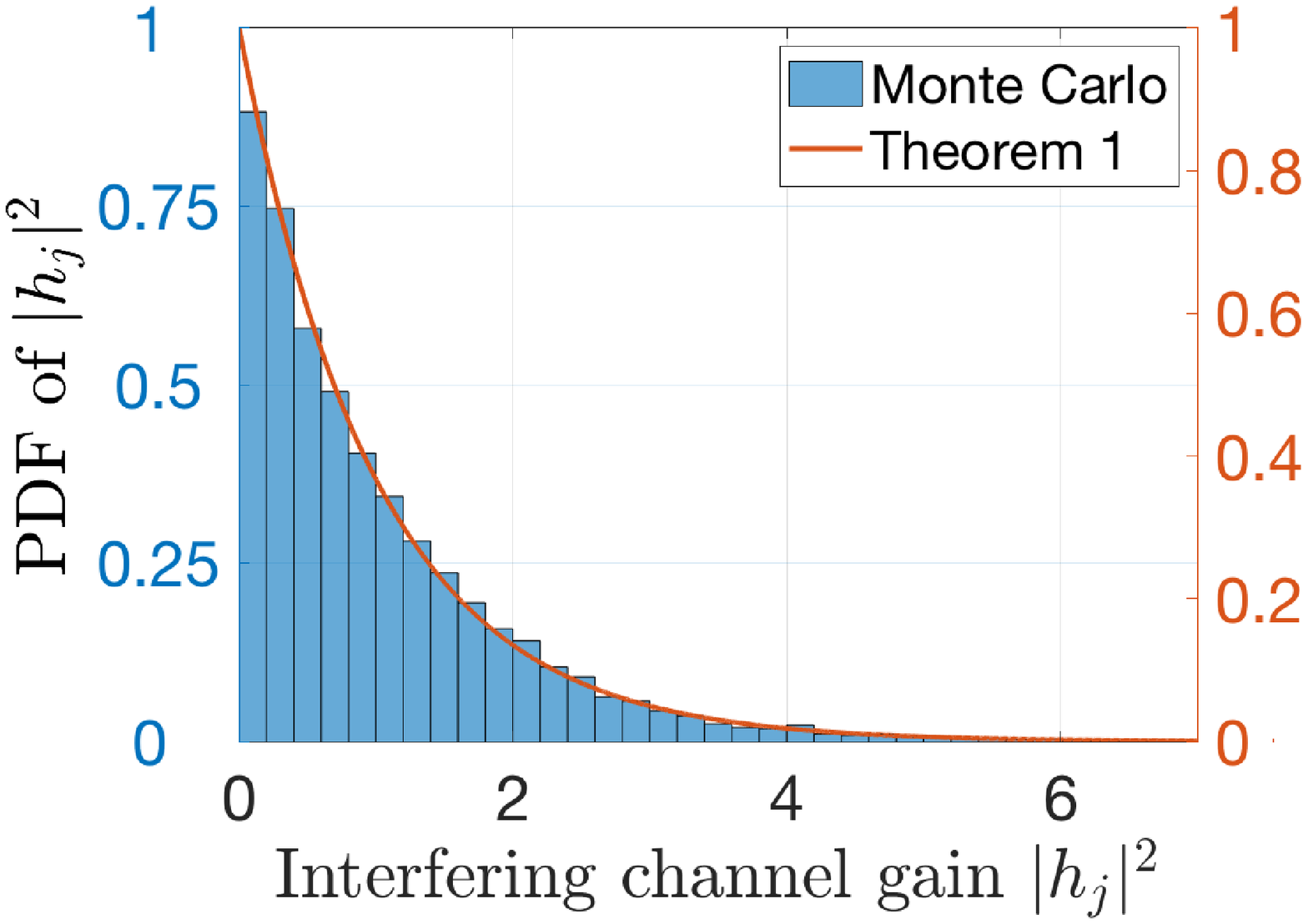}
        \label{pwr2}
    }
    \caption{ \ac{PDF} of the interfering channel power.}		
    \label{inter-pwr}
    \vspace{-0.4cm}
\end{figure}
\begin{figure}[!t]	
\vspace{-0.4cm}
    \centering
    \subfigure[\black{AU altitude $h_d=\SI{90}{m}$}]
    {
        \includegraphics[width=1.6in]{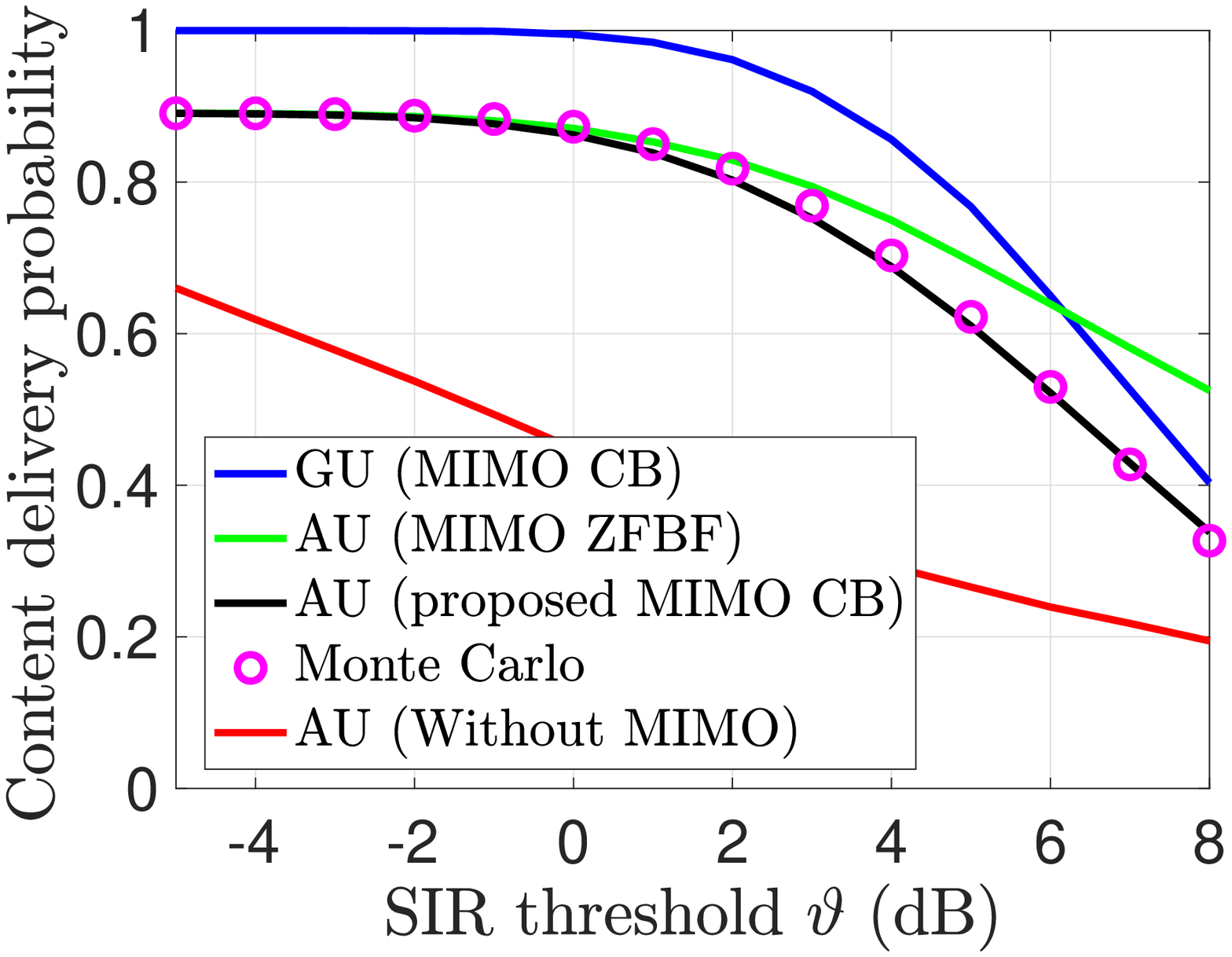}		
        \label{cov_prob_vs_theta}
    }
    \subfigure[\black{$\vartheta=\SI{5}{dB}$, $\lambda=\SI{50}{km^{-2}}$}]
    {
        \includegraphics[width=1.6in]{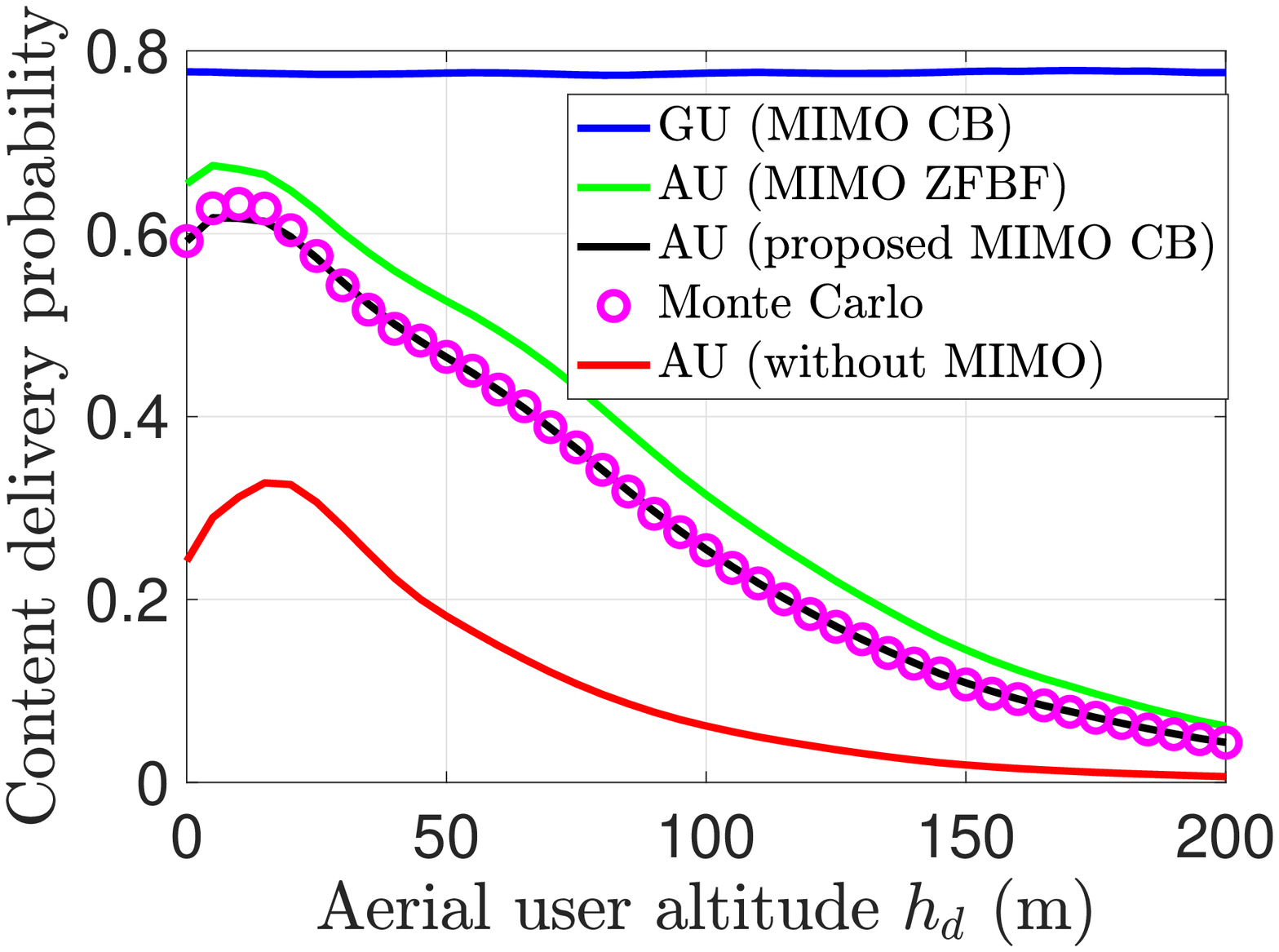}	
        \label{cov_prob_vs_hd}
    }
    \caption{Effect of \ac{SIR} threshold and \ac{AU} altitude (\black{$h_{\rm BS}=\SI{30}{m}$}).}		
    \label{cov_prob_vs_theta_hd}
    \vspace{+0.2cm}
\end{figure}
\begin{proof}
Please see Appendix A.
\end{proof}
\begin{remark}
\black{\rm The main merit of this representation, i.e., $\Pb_{{\rm c}|r}^v = \lVert e^{\boldsymbol{T}_{M_v}}\rVert_1$, is that it leads to valuable system insights. For example, the impact of the shape parameter $M_v=Mm_v$ on the intended channel gain $h_{iK}\sim \Gamma(Mm_v,\eta/m_v)$, which is typically related to the antenna size and the Nakagami fading parameter $m_v$, is clearly illustrated by this finite sum representation (\ref{prob-y2}). Although it is not tractable to obtain closed-form expressions for $t_{k}$ (the entries populating $\boldsymbol{T}_{M_v}$), special cases of interest, e.g., \ac{LoS} or \ac{NLoS} communications, can lead to closed-form expressions, following \cite{8490204}.}
%
%
%
\end{remark}

\begin{remark}
\black{\rm When $K=1$, only the AU is scheduled, i.e., maximal ratio transmission (MRT) beamforming. For MRT, the interfering channel gain is $\Gamma(1,\eta)$. Interestingly, this interfering channel gain is reduced as opposed to the typical Nakagami channel gain $\Gamma\big(m_l,\frac{\eta}{m_l}\big)$ when there is neither precoding nor MIMO transmission.}
\end{remark}

%
%

\vspace{-0.3cm}
\section{Numerical Results} 
For our simulations, we consider a network having the following parameters, \black{unless otherwise specified}. \black{The number of antennas per sector is set to} $M=32$. We also set $K=4$, $\lambda=\SI{1}{km^{-2}}$, \black{$h_{{\rm BS}} = \SI{55}{m}$}, $h_g=\SI{1}{m}$, $\alpha_l = 2.09$, $\alpha_n = 3.75, a = 0.6, \nu = 500 \SI{}{km^{-2}}, c = 25, \black{\vartheta=\SI{10}{dB}}, A_{l}=\SI{-41.1}{dB}, A_{n}=\SI{-32.9}{dB}, G_m=\SI{10}{dB}, G_s=\SI{-3.01}{dB}, m_n=1, m_l=3, \eta=1, \sigma^2=1$, $\theta_B = 45^{\circ}$, $\theta_t = 30^{\circ}$. 

In Fig.~\ref{inter-pwr}, we verify the accuracy of the obtained \ac{PDF} of interfering channel gain $|h_j|^2$ (Table \ref{ch5:table1}: row 4) in Theorem \ref{ch5:theorem1}. The figure shows that the derived \ac{PDF} is quite accurate when $M$ is sufficiently large as in Fig.~\ref{pwr2}, while for small $M$ in Fig.~\ref{pwr1}, it still provides a reasonable approximation.

Fig.~\ref{cov_prob_vs_theta_hd} compares the \ac{SCDP} of \black{\acp{AU} with and without MIMO beamforming to GUs}. Fig.~\ref{cov_prob_vs_theta} plots the \ac{SCDP} as a function of the \ac{SIR} threshold $\vartheta$ for the \ac{AU} and the \acp{GU}. Clearly, the achievable performance of \acp{GU} considerably outperforms that of an \ac{AU}. This is because \acp{GU} have a superior propagation environment, driven by the down-tilted BS antennas in the desired signal side, and the \ac{NLoS} interfering links. 
\black{However, Fig.~\ref{cov_prob_vs_theta_hd} also shows that the \ac{SCDP} for the \ac{AU} served by MIMO \ac{CB} significantly outperforms that of the \ac{AU} served by single-antenna \acp{GBS}. Moreover, although the \ac{ZFBF} technique outperforms our proposed CB approach, the low complexity of \ac{CB} and its  associated performance gain over traditional single-antenna \acp{GBS} make it a good candidate to serve \acp{AU}.}  
Fig.~\ref{cov_prob_vs_hd} shows the effect of \ac{AU} altitude on the AU performance, with that of \acp{GU} plotted for comparison. \black{Fig.~\ref{cov_prob_vs_hd} shows that the \ac{AU} \ac{SCDP} (for all transmission schemes) gradually increases with $h_d$ up to a maximum value due to the larger \ac{LoS} probability, before it decreases again due to the stronger \ac{LoS} interference and higher large scale fading.}

\begin{figure}[!tbp]	
    \vspace{-0.6cm}
    \centering
    \subfigure[AU altitude $h_d=\SI{30}{m}$]			
    {
        \includegraphics[width=1.62in]{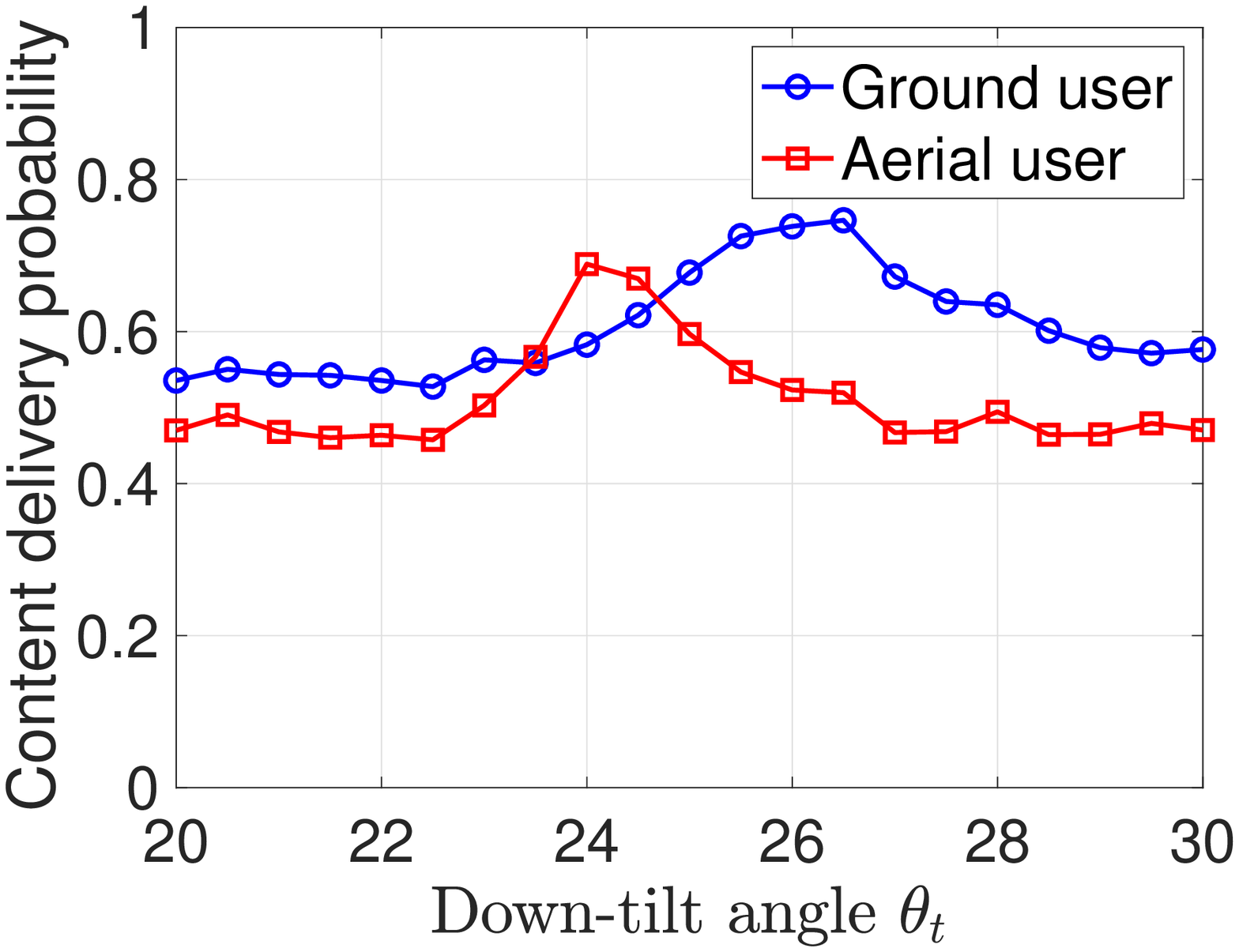}		
        \label{ch5:tilt1}
    }
    \subfigure[AU altitude  $h_d=\SI{80}{m}$]		
    {
        \includegraphics[width=1.62in]{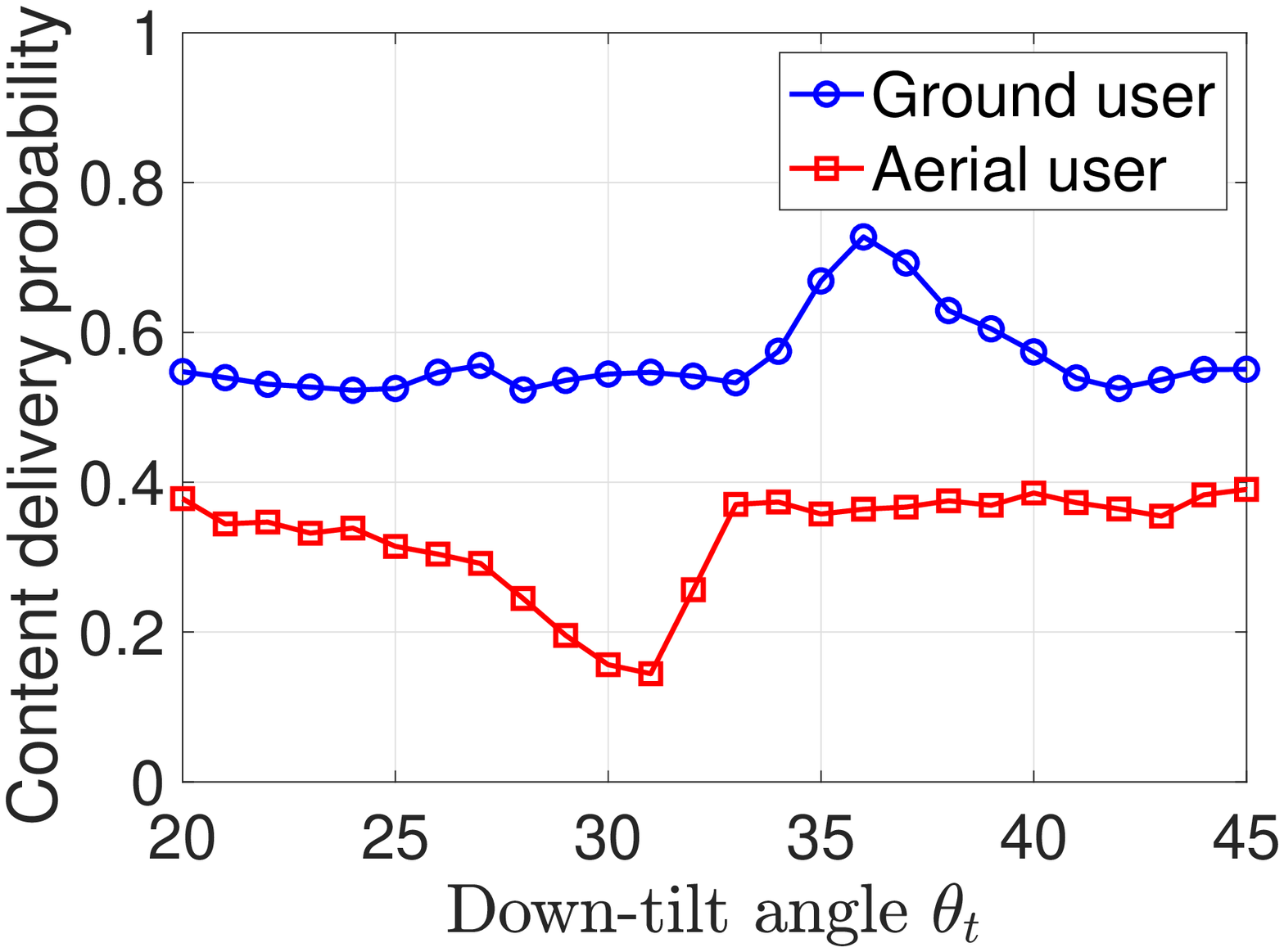}
        \label{ch5:tilt2}
    }
    \caption{Effect of antenna down-tilt angle.}   
    \label{ch5:tilt}
    \vspace{-0.4cm}
\end{figure}

Fig.~\ref{ch5:tilt} illustrates the effect of the down-tilt angle $\theta_t$ on the performance of both the \ac{AU} and the \acp{GU}, for different \ac{AU} altitudes. As illustrated in Fig.~\ref{ch5:tilt1}, for $h_d<h_{{\rm BS}}$, the performance of the \ac{AU} is maximized at certain $\theta_t$, and beyond that it starts to degrade. However, for \acp{GU}, their performance is maximized at a higher $\theta_t$. \black{Hence, adjusting the antennas' down-tilt angle yields a tradeoff between the performance of AUs and GUs owing to the difference in their altitudes}. For $h_d>h_{{\rm BS}}$ in Fig.~\ref{ch5:tilt2}, the \ac{SCDP} of the \ac{AU} first decreases with $\theta_t$ to a minimum value, and then it increases again. This finding can be explained as follows: when $\theta_t$ is small, an \ac{AU} at an altitude $h_d>h_{{\rm BS}}$ can be served from the main lobe of tagged BS while also experiencing high interference from the main lobe of other interfering BSs.  Gradually, as $\theta_t$ increases, the worst performance is observed when the \ac{AU} is no longer served from the main lobe of tagged BS antennas while still experiencing high interference from the main lobe of other BSs. Finally, for very large $\theta_t$, both intended and interfering signals stem from the side-lobes, and hence the performance is improved again. 
In Fig.~\ref{ch5:cov_vs_KM},  we show the prominent effect of the number of scheduled users $K$ and the number of antennas $M$ on the network performance. Fig.~\ref{cov_vs_K} shows that the \ac{SCDP} monotonically decreases for both AU and GU as $K$ increases due to the effect of stronger interference. However, it is noticeable that increasing $K$ highly degrades the \ac{AU} performance compared to that of \acp{GU}. This stems from the fact that \acp{AU} are more sensitive to interference, which often exhibits \ac{LoS} component.  
\black{In Fig.~\ref{SE_vs_K}, we show the system spectral efficiency (SE) versus the number of scheduled users $K$. In this figure, $K=1$ means that only the AU is scheduled. Evidently, the system SE increases as $K$ increases, which proves that spatially multiplexing one \ac{AU} with the \acp{GU} significantly improves the system SE.}
Fig.~\ref{cov_vs_M} shows that increasing the number of antennas $M$ improves the \ac{SCDP} for both users with nearly the same rate. 
 

\begin{figure}[!t]	
    \vspace{-0.5cm}
    \centering
    \subfigure[Number of antennas $M=32$]
    {
        \includegraphics[width=1.04in]{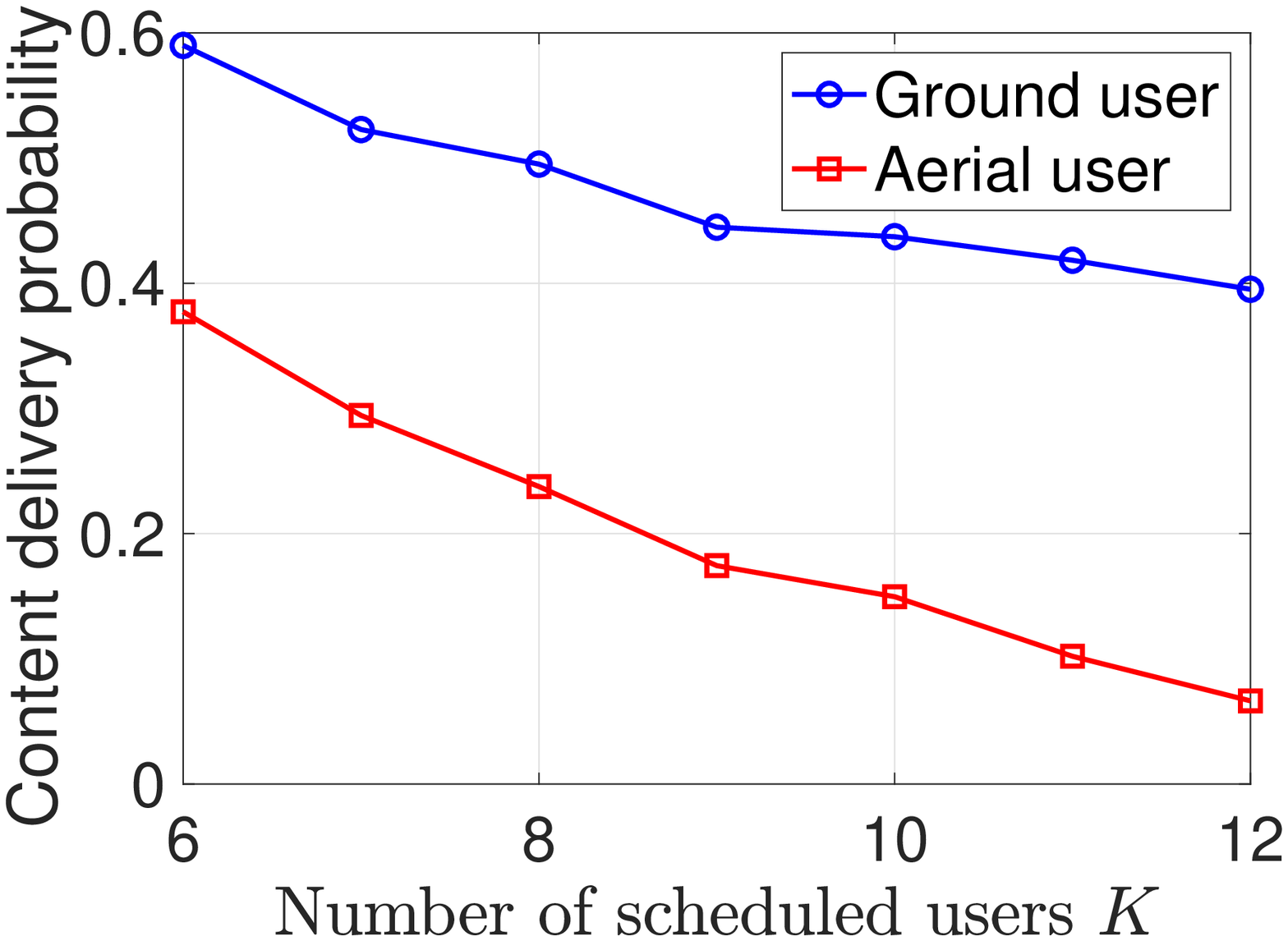}		
        \label{cov_vs_K}
    }
    \subfigure[Number of antennas $M=32$]
    {
        \includegraphics[width=1.04in]{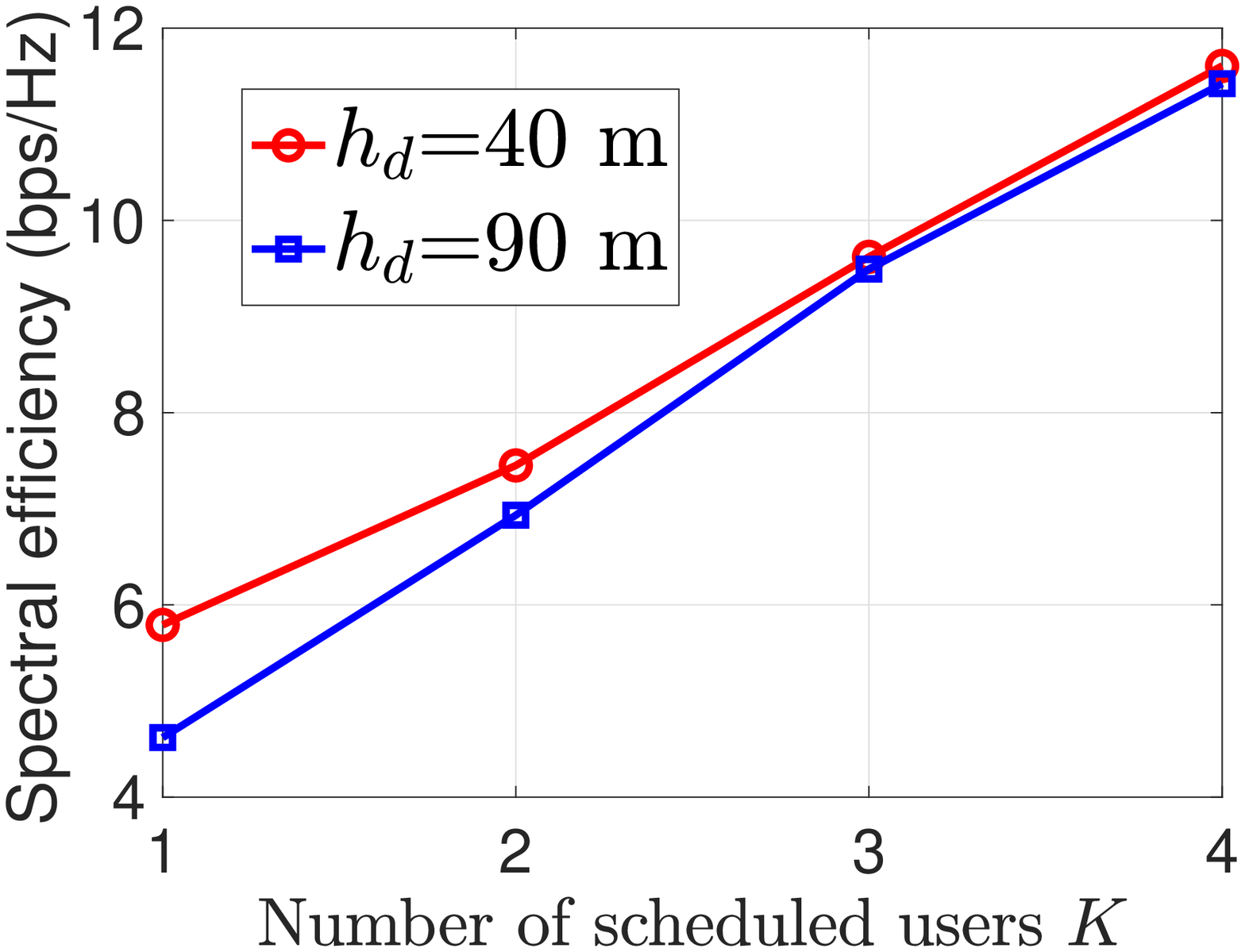}		
        \label{SE_vs_K}
    }
    \subfigure[Number of users $K=4$]
    {
        \includegraphics[width=1.04in]{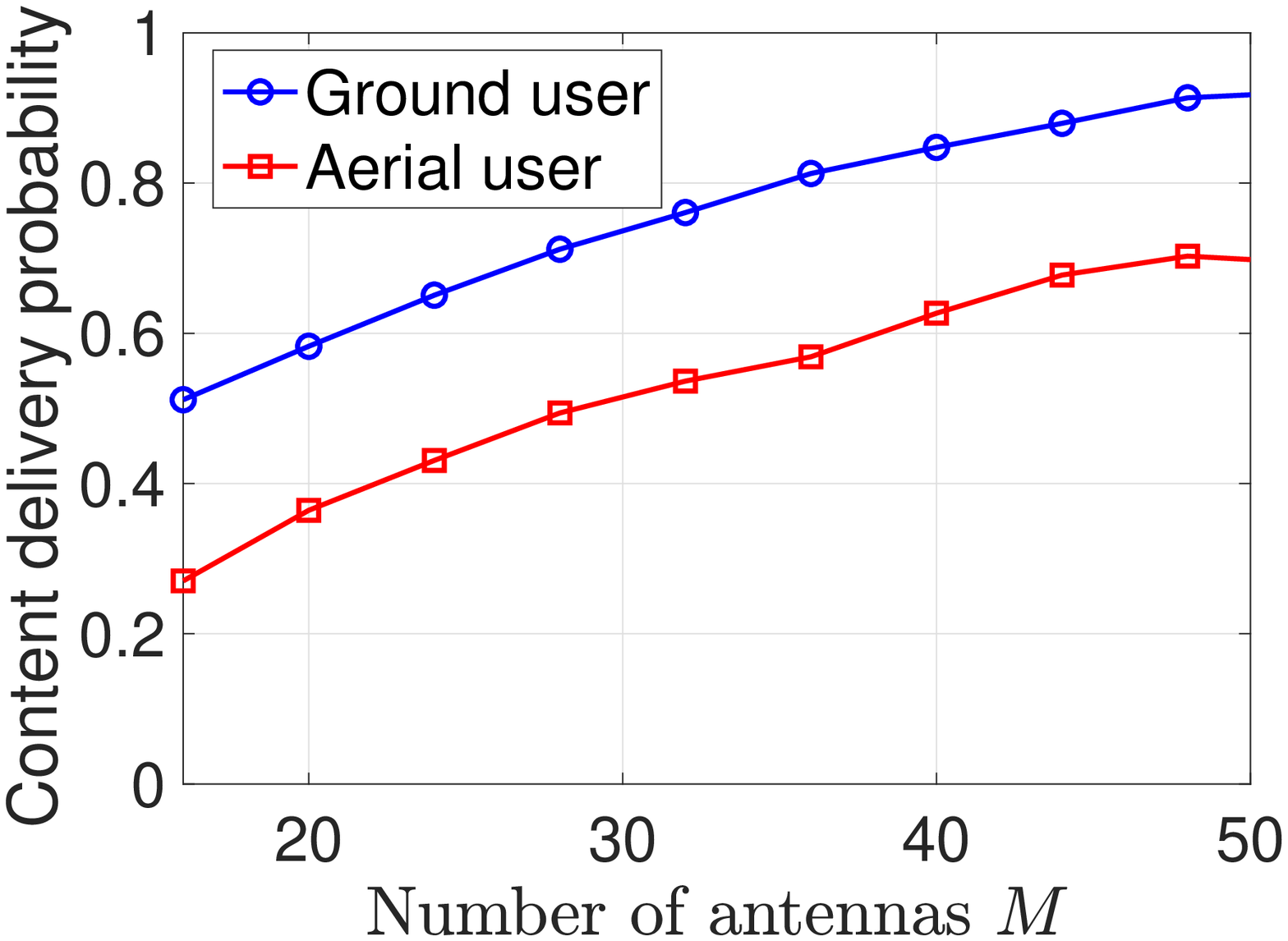}
        \label{cov_vs_M}
    }
    \caption{Effect of the number of antennas and the number of scheduled users.}		
    \label{ch5:cov_vs_KM}
    \vspace{-0.3cm}
\end{figure}

\vspace{-0.3 cm}
\section{Conclusion} 
In this paper, we have proposed a novel \ac{CB} framework for spatially multiplexing \acp{AU} and \acp{GU}. In order to analytically characterize the \ac{SCDP}, we have derived the gain of intended and interfering channels. \black{We have shown that  exploiting \ac{CB} from massive  \ac{MIMO}-enabled BSs to spatially multiplex one AU and multiple  \acp{GU} substantially improves the performance of the \ac{AU}}, in terms of \ac{SCDP}. We have then shown that the down-tilt of the BS antennas can lead to a tradeoff between the performance of AUs and GUs only if the AU's altitude is below the BS height. Simulation results have shown the various properties of cellular communications when AUs and GUs co-exist.\footnote{Creating communication protocols for secure content delivery for networks of \acp{UAV} using, e.g., blockchain technology, can be a potential subject for future investigation \cite{8662639,baza2018blockchain,baza2019b,baza2018blockchain1}.}  
\vspace{-0.3 cm} 
\begin{appendix}
 \vspace{-0.1 cm}
\black{The \ac{SCDP} is defined as the probability of downloading content with a received $\sir$ higher than a target threshold $\vartheta$, i.e.,}
\begin{align}				
\Pb_{{\rm c}|r}^v &= \Pb\Big(\frac{ \frac{P_t}{K} \zeta_v(r) \big|\boldsymbol{w}_{i1} \boldsymbol{h}_{i1} \big|^2 }{I}>\vartheta\Big) 
=\Pb\Big(\Big|\boldsymbol{w}_{i1} \boldsymbol{h}_{i1} \Big|^2>
\frac{\vartheta K}{P_t \zeta_v(r)} I\Big) 
\nonumber  \\
  \label{prob-y2}
&\overset{(a)}{=} \Eb_{I} \Bigg[ \sum_{i=0}^{M_v-1}   \frac{s_v^i }{i!}I^i
e^{-s_vI} \Bigg]
\overset{(b)}{=}  \sum_{i=0}^{M_v-1}  \frac{(-s_v)^i}{i!} \Lc_{I|r}^{(i)}(s_v), 
\end{align}					
where $I=I_{\text{in}}+I_{\text{out}}$, \black{(a) follows from $|\boldsymbol{w}_{i1} \boldsymbol{h}_{i1}|^2 \sim \Gamma(M_v,\frac{\eta}{m_v})$, and (b) follows from the Laplace transform of interference, along with the assumption of independence between the intra- and inter-cell interference.}  
Next, we derive the Laplace transform of interference $\Lc_{I|r}(s_v)$ from:
\begin{align}
& = \Eb_{I} \Big[e^{-s_v I} \Big] = 
  \mathbb{E}_{h_{iK}} e^{-s_v h_{iK}  P(r)^2} 
  \mathbb{E}_{\Phi}  \prod_{j \in \Phi^o}\mathbb{E}_{h_{jK}} e^{-s_v h_{jK}  P(u_j)^2}
  \nonumber \\
  &\overset{(a)}{=}    \Big(1 + s_v \eta P_v(r)^2 \Big)^{-(K-1)}
      e^{-2\pi  \lambda \mathbb{E}_{h_{jK}}  \int_{\nu=r}^{\infty}
    \Big(1 -       {\rm exp}\big(-s_v h_{jK} P(\nu)^2\big) \Big)\nu\dd{\nu}} 
      \nonumber \\
      &\overset{(b)}{=}    e^{-(K-1) {\rm log}(1 + s_v \eta P_v(r)^2)}
      e^{-2\pi  \lambda \mathbb{E}_{h_{jK}}  \int_{\nu=r}^{\infty}
    \Big(1 -       {\rm exp}\big(-s_v h_{jK} P(\nu)^2\big) \Big)\nu\dd{\nu}} 
      \nonumber \\
&\overset{(c)}{=}    e^{-(K-1) {\rm log}(1 + s_v \eta P_v(r)^2)}
\times 
\nonumber \\
      &e^{-2\pi  \lambda \int_{\nu=r}^{\infty}
    \big(1 -       
    \Pb_{l}(\nu)\delta_l(\nu,s_v) - \Pb_{n}(\nu)\delta_n(\nu,s_v) 
    \big) \nu\dd{\nu}}  \overset{}{=}   e^{\varpi(s_v)}, 
    \nonumber 
\end{align}
\black{where (a) follows from $h_{iK} \sim \Gamma(K-1,\eta)$ and the \ac{PGFL} of \ac{PPP} $\Phi$ \cite{haenggi2012stochastic}}. (b) follows from the fact that $x=e^{{\rm log}(x)}$, and (c) follows since $h_{jK} \sim \Gamma(K,\eta)$. 
In \cite{8490204}, it is proved that  $  \sum_{i=0}^{M_v-1}  \frac{(-s_v)^i}{i!} \Lc_{I|r}^{(i)}(s_v)  = \sum_{i=0}^{M_v-1}  p_i$, with $p_i =\frac{(-s_v)^i}{i!} \Lc_{I|r}^{(i)}(s_v)$ computed from the recursive relation:
$p_i = \sum_{l=0}^{i-1}\frac{i-l}{i} p_l t_{i-l}$, where $t_{k} =\frac{(-s_v)^{k}}{k!} \varpi^{(k)}(s_v)$. After some algebraic manipulation, $\Pb_{{\rm c}|r}^v$ can be expressed in a compact form $\Pb_{{\rm c}|r}^v = \lVert e^{\boldsymbol{T}_{M_v}}\rVert_1$ as in \cite{8490204}. \black{In summary, we first derive the conditional log-Laplace transform $\varpi(s_v)$ of the aggregate interference. Then, we calculate the $n$-th derivative of $\varpi(s_v)$ to populate the entries $t_n$ of the lower triangular Toeplitz matrix $\boldsymbol{T}_{M_v}$. The conditional \ac{SCDP} can be then computed from $\Pb_{{\rm c}|r}^v = \lVert e^{\boldsymbol{T}_{M_v}}\rVert_1$.}

\end{appendix}
\vspace{-0.09 cm}
\bibliographystyle{IEEEtran}
\bibliography{bibliography}
\end{document}